\documentclass[letterpaper, 10 pt, conference]{ieeeconf}  

\IEEEoverridecommandlockouts                              
\usepackage{epsfig} 
\usepackage{amssymb}  
\usepackage{upgreek}
\usepackage{array}
\usepackage{caption}
\usepackage[fleqn]{mathtools}
\usepackage{enumerate}
\usepackage{caption}
\usepackage{upgreek}
\usepackage{color}
 
 \allowdisplaybreaks
 
\mathtoolsset{showonlyrefs}

\usepackage{graphicx}
\graphicspath{{./figures/}}
\usepackage{epsfig}
\usepackage{subfig}
\usepackage{tikz}
\usetikzlibrary{decorations.pathmorphing}
\usetikzlibrary{arrows,automata,positioning,shapes}
\usetikzlibrary{decorations.pathmorphing}
\tikzset{snake it/.style={decorate, decoration=snake}}

\newtheorem{assumption}{Assumption}
\newtheorem{lemma}{Lemma}
\newtheorem{theorem}{Theorem}

\usepackage{filecontents}

\newcommand{\Rmax}{\mathrm{R}_{\text{max}}}
\newcommand{\n}{\mathrm{n}}
\renewcommand{\r}{\mathrm{r}}
\renewcommand{\a}{\mathrm{a}}
\newcommand{\s}{\mathrm{s}}
\newcommand{\R}{\mathrm{R}}
\newcommand{\Ri}{\mathrm{R}_i}
\newcommand{\w}{\upomega}
\renewcommand{\b}{\mathrm{b}}
\newcommand{\K}{\mathrm{K}}
\renewcommand{\u}{\mathrm{u}}
\newcommand{\z}{\mathrm{z}}
\usepackage{cite}
\title{\LARGE \bf
Stability of Leaderless Resource Consumption Networks}
\author{Sebastian F. Ruf$^{1}$, Matthew T. Hale$^{2}$, Talha Manzoor$^{3}$, Abubakr Muhammad$^{4}$
\thanks{$^{1}$ S.F. Ruf is with the Department of Electrical and Computer Engineering, Georgia Institute of Technology, Atlanta, GA. Email:{\tt\small ruf@gatech.edu}}%
\thanks{$^{2}$ M.T. Hale is with the Department of Mechanical and Aerospace Engineering,
University of Florida, Gainesville, FL. Email: {\tt\small matthewhale@ufl.edu}}%
\thanks{$^{3}$ T. Manzoor is with the Department of Electrical Engineering, Namal College, Mianwali, Pakistan. Email: {\tt\small talha@namal.edu.pk}}%
\thanks{$^{4}$ A. Muhammad is with the Department of Electrical Engineering, Lahore University of Management Sciences (LUMS), Lahore, Pakistan. Email: {\tt\small abubakr@lums.edu.pk}}%
}%

\begin{document}

\maketitle
\thispagestyle{empty}
\pagestyle{empty}

\begin{abstract}
In this paper, we study the global stability properties of a multi-agent model of natural resource consumption that balances ecological and social network components in determining the consumption behavior of a group of agents. The social network is assumed to be leaderless, a condition that ensures that no single node has a greater influence than any other node on the dynamics of the resource consumption. It is shown that any network structure can be made leaderless by the social preferences of the agents. 
The ecological network component includes a quantification of each agent's environmental concern, which
captures each individual agent's threshold for when a resource becomes scarce. We show that leaderlessness
and a mild bound on agents' environmental concern are jointly sufficient for global
asymptotic stability of the consumption network to a positive consumption value, indicating
that appropriately configured networks can continuously consume a resource without
driving its value to zero. 
The behavior of these leaderless resource consumption networks is verified in simulation. 
\end{abstract}

\section{Introduction}
In the face of an ever-changing natural climate, understanding the behavior of renewable natural resources and the impact of human consumption on those resources is important for ensuring long term resource consumption \cite{berkes2008navigating,Liu_1513}. Modeling of natural resources allows the prediction of consumption behavior and offers valuable insights into the relationship between various system components. In this paper we study network structure and resource consumption. Of particular interest is the equilibrium behavior of these models, as equilibria can help describe the long term sustainability of natural resources \cite{SolowRobertM2000Gt:a}. The discussion of long term system behavior must be preceded by an understanding of the stability properties of the system. 

This paper focuses on the study of an agent-based model of natural resource consumption previously introduced and studied in \cite{manzoor2016game, manzoor2017structural, manzoor2018learning}. This model captures insights from the social sciences on the consumption behavior of humans in a form that can be analyzed mathematically.  Past work on this model has sought to understand the behavior of the model and has considered stability of this model in the two agent case. This paper extends the consideration of stability to consider $\n$ agents interacting over a network. 

The overall system consists of an ecological sub-model, which describes the dynamics of the resource, and a social sub-model, describing the dynamics of the agents' consumption. The ecological sub-model is based on the Gordon Schaefer model, which represents a class of well studied systems associated with dynamic processes in population biology, ecological economics and other related disciplines \cite{gordon1954economic,perman2003natural}. The stability of the Gordon-Schaefer model, as well as similar logistic growth models, has been studied extensively in isolation from the social processes which drive human consumption behavior \cite{hsu1978global,freedman1985global,chiu1999lyapunov,hsu2005survey}.  

The social component of the model describes the cognitive decision making process of the agents regarding change in their consumption. This process is influenced both by the state of the resource and the consumption of other neighboring agents. The influence of the agents on each others' consumption is similar to how agents influence each other in mathematical models of opinion formation \cite{abelson1964mathematical} 
and consensus in cooperative multi-agent systems \cite{mesbahi2010graph}. The dependence of the agents' resource consumption on the state of the resource appears as an exogenous factor or time-varying bias in the overall dynamics (see \cite{mirtabatabaei2014eulerian, russo2010global} for similar models). An important component of the social process is the underlying social network structure, which greatly influences the ability of a community to successfully manage its natural resources \cite{ostrom2015governing,videras2013social}. 

This paper studies the behavior of a consumption network under the assumption that it is leaderless: that one agent will not drive the social network component of the model more than any other agent in the network. This assumption allows an aggregation of individual state nodes \cite{manzoor2017structural}, facilitating an understanding of the system-level behavior. Discussing a consuming population in aggregate is a common tool for the study of resource consumer social networks\cite{prell2011social} and allows for the design of actions taken to change behavior, which often happen at the community level\cite{videras2013social}. The leaderless assumption, as will be shown, also captures a wide array of systems and a rich class of stable system behaviors. 

The rest of the paper is organized as follows.
Section~\ref{sec:model} introduces the consumption model and discusses its properties.  
Section \ref{sec:leaderless} discusses the leaderless condition and presents a Lyapunov based proof for global stability of the system. In Section \ref{sec:sim} the behavior of leaderless systems is studied in simulation, with a discussion in Section \ref{sec:disc}. The paper concludes in Section~\ref{sec:conc}.

\section{System Dynamics}
\label{sec:model}
This section presents the dynamics governing the resource quantity and consumer behavior in the coupled socio-ecological system. We first discuss each sub-model and then give an
aggregate leaderless consumption model. 

\subsection{The Ecological Sub-model}
The ecological component of the system is assumed to consist of a single renewable resource with quantity at time $\tau$ represented by $R(\tau)$. In the absence of consumption, the resource grows at intrinsic growth rate $\r$ and saturates at carrying capacity $\Rmax$. The resource is connected to a consuming population consisting of $\n$ individuals. Each individual can harvest the resource by exerting consumption effort $e_i(\tau)$, where $i \in \{1,\dots, \n\}$ represents a single consumer. The resource dynamics are assumed to follow the standard Gordon-Schaefer model \cite{perman2003natural} with catch coefficient equal to one, which is given as
\begin{align}
\label{eq:res}
	\frac{d R(\tau)}{d \tau} = \r R(\tau) \left( 1 - \frac{R(\tau)}{\Rmax} \right) - R(\tau) \sum_{i=1}^{\n} e_i(\tau).
\end{align}

\subsection{The Social Sub-model}
The social sub-model is based on Festinger's theory of social comparison processes \cite{festinger1954theory}, which postulates that human beings evaluate their decisions, opinions and abilities by reflecting on both objective and social information. In the context of natural resource consumption, objective information corresponds to the state of the resource and social information corresponds to the consumption of other socially connected individuals \cite{mosler2003integrating}. To balance between objective and social information, the change in consumption effort of the agent is given as a weighted sum of both ecological and social factors. 

The ecological factor for consumer $i$ is given by $\displaystyle R(\tau)-\Ri$, where $\Ri \in \mathbb{R}$ represents the perceived scarcity threshold of $i$, below which agent $i$ considers the resource to be scarce, and above which she considers it to be abundant. The ecological factor is weighed by $\a_i \in (0,\infty)$, which represents the set of factors to which agent $i$ attributes the state of the natural resource. An ecological attribution $\a_i \rightarrow 0$ represents a consumer that attributes the state of the resource entirely to the actions of the consuming society (including the agent itself), while increasing values of $\a_i$ correspond to the individual attributing the current state of the resource to natural causes (droughts, wildfires, heavy rain, etc).

The ecological factor is balanced by a social component, given by $\displaystyle \sum_{i=1}^{\n} \w_{ij} (e_j(\tau)-e_i(\tau))$, which is the difference between $i$'s consumption and that of the other socially connected consumers in the population. The graph connectivity is captured by $\w_{ij}\geq0$ which is the strength of the social tie directed from $j$ to $i$. We assume that $\displaystyle \sum_{j=1}^{\n} \w_{ij} = 1$ and $\w_{ii} = 0 \, \forall \, i\in \{1,\dots, \n \}$. The social factor is weighed by $\s_i \in (0,\infty)$, the social-value orientation of $i$. Social-value orientations $\s_i \rightarrow 0$ represent extremely non-cooperative individuals, which will ignore the actions of their network neighbors. Conversely,  increasing values of $\s_i$ correspond to increasingly cooperative individuals. 

Combining the social and ecological component gives the dynamics of the consumption effort for consumer $i$ as
{\setlength{\mathindent}{0pt}
\begin{align}
\label{eq:eff}
	\frac{d {e}_i(\tau)}{d \tau} = \a_i (R(\tau)-\Ri) + \s_i \sum_{j=1}^{\n} \w_{ij} (e_j(\tau)-e_i(\tau)),
\end{align}}
where the ecological and social factors have been weighed in accordance with findings in social psychological research on consumer behavior \cite{manzoor2016game}. In particular, individuals that attribute blame to natural causes tend to give more importance to ecological information and vice versa.  Similarly, cooperative individuals are more concerned with maximizing equality in consumption than non-cooperative ones, and as such will be further influenced by the social factor (see \cite{manzoor2016game} and included references for more details).

\subsection{Non-dimensionalized Socio-Ecological System}
In order to reduce the dimensionality of the parameter space, the system given by Eq. \eqref{eq:res} and \eqref{eq:eff} is non-dimensionalized. The process of non-dimensionalization has an added benefit of allowing comparison between system parameters. The dynamics of the non-dimensionalized state of the resource $\displaystyle x=\frac{R(\tau)}{\Rmax}$ and the non-dimensionalized consumption $\displaystyle y_{i}=\frac{e_{i}}{\r}$ are given as follows
{\setlength{\mathindent}{0pt}
\begin{align}
\begin{split}
\label{eq:ses}
	\dot{x} &= (1-x)x - x\sum_{i=1}^\n y_i,\\
	\dot{y}_i &= \b_i \left( (1-\upnu_i) (x-\uprho_i) + \upnu_i \sum_{j=1}^{\n} \w_{ij} (y_j - y_i) \right),
\end{split}
\end{align}}
where $i \in \{1,\dots, \n\}$, $$\displaystyle \b_i = \frac{\a_i \Rmax + \r \s_i}{\r^2}, ~~~\displaystyle \upnu_i = \frac{\r \s_i}{\a_i \Rmax + \r \s_i},$$ and the derivatives $\dot{x}$ and $\dot{y}_i$ are taken with respect to the non-dimensional time $t = \r \tau$. The non-dimensionalized threshold $\displaystyle \uprho_i=\frac{\Ri}{\Rmax}$ is called the environmentalism of $i$. The parameter $b_i$ is the sensitivity of $i$, which represents $i's$ openness to change in her consumption. The final parameter, $\upnu_i$, is called the socio-ecological relevance of $i$ and represents the importance that $i$ gives to social information relative to ecological information in the process of changing consumption behavior.

\subsection{Influence and Leadership}
The consumption of $i$ is influenced by the consumption of all other agents that are socially connected to her. This notion of connectivity is captured in Eq. \eqref{eq:ses} via the parameters $\upomega_{ij}$, which denote the strength of the social tie directed from $j$ to $i$. If $\w_{ij}=0$ this implies that there is no social link from $j$ to $i$, allowing the collection of $\w_{ij}$'s to specify the topology of the underlying social network. The aggregate influence of the rest of the agents on $i$'s consumption is given by $\sum_{j=1}^{\n} \b_i \upnu_i \w_{ij}$ and is called the in-influence of $i$. The aggregate influence that $i$ exerts on the other agents in the network is given by $\sum_{j=1}^{\n} \b_j \upnu_j \w_{ji}$ and is called the out-influence of $i$. The difference between the out-influence and the in-influence is called the net-influence of $i$ and determines the role of $i$ in the network as a leader (positive net-influence), a follower (negative net-influence) or neutral (zero net-influence). In this paper, we consider cases in which all agents in the network are neutral.

\section{Global Asymptotic Stability of Leaderless Networks}
\label{sec:leaderless}
In this section, two assumptions on the network and parameters are introduced before transforming the non-dimensionalized dynamics in
Eq.~\eqref{eq:ses} into a form more amenable to stability analysis. 
Then the equilibrium point of the transformed dynamics is computed and a coordinate
shift is applied to move this equilibrium point to the origin. Finally global
asymptotic stability of the origin is proven, which implies stability
of the non-dimensionalized dynamics in Equation~\eqref{eq:ses}. 

\subsection{Transformed Leaderless Aggregate Dynamics}
In this section, two assumptions on the model parameters are introduced.  

\begin{assumption} \label{as:leaderless}
The network under consideration is leaderless, i.e.,
\begin{equation}
\sum_{j=1}^{\n}\big(\w_{ij}\b_i\upnu_i - \w_{ji}\b_j\upnu_j\big) = 0
\end{equation}
for all $i \in \{1, \ldots, \n\}$. \hfill $\triangle$
\end{assumption}

Note that this assumption is the same used in defining homogeneous consumer networks in \cite{manzoor2017structural}. 
Below, the network-level dynamics will be derived by considering the new state variables
\begin{equation} 
	\qquad\qquad z = \log x \quad \textnormal{and } \quad	u = \sum_{i=1}^{\n} y_i.
\end{equation}

In doing so the following assumption, which bounds the maximum possible value of $\uprho_i$ for each agent, will be enforced.

\begin{assumption} 
\label{as:rhoi}
	For all $i \in \{1,\dots,\n\}$, $\uprho_i \in (0, 2)$. \hfill $\triangle$
\end{assumption}

Because $\uprho_i$ is the normalized value of $\Ri$, Assumption~\ref{as:rhoi} implies that $\Ri \in (0, 2\Rmax)$. This is therefore
a rather weak 
assumption as few agents are expected to have $\Ri >\Rmax$
because this implies that agent $i$'s 
scarcity threshold
is \emph{larger} than the resource carrying capacity $\Rmax$. 

\subsection{Leaderless Networks}
Before deriving the modified network level dynamics of the system, this section considers Assumption~\ref{as:leaderless} in more detail. First, we consider the existence of a leaderless dynamic.  
\begin{lemma}\label{lem:lead_ex}
For any set of network weights $\upomega_{ij}$, there exists a set of social orientations that renders the network leaderless. 
\end{lemma}
\begin{proof}
Consider the matrix of edge weights,\\
\resizebox{\linewidth}{!} {%
$W=\begin{bmatrix}
-\left(\sum_{j=2}^{n}\w_{1j}\right)& \w_{21}& \dots &\w_{n1} \\
\w_{12}&-\left(\sum_{\substack{j=1\\j\neq2}}^{n}\w_{2j}\right)& \dots&\w_{n2}\\
\vdots&\vdots&\ddots&\vdots\\
\w_{1n}&\w_{2n}&\dots&-\left(\sum_{j=1}^{n-1}\w_{nj}\right)
\end{bmatrix}.$ \\
}
\hspace{.1cm}
The leaderless condition is equivalent to $W$ having a non-trivial null space. Notice that the matrix $W^{T}$ has rows that sum to $0$, which implies that the vector $1_{n}$ is an eigenvector with eigenvalue $0$. As $W$ and $W^{T}$ have the same eigenvalues \cite{horn1990matrix}, $W$ also has an eigenvalue at $0$ and therefore a non-trivial null space. Then there must be a vector of normalized social orientation $bv\in\mathrm{null}(W)$ and such a vector renders the network leaderless. 
\end{proof}
Figure~\ref{fig:led_ex} shows two examples of a leaderless network for a set of uniform weights on a line graph and a cycle graph. Lemma~\ref{lem:lead_ex} shows that any graph, including those commonly found in complex networks such as scale free \cite{albert2002statistical} and small world \cite{watts1998collective} networks, can be rendered leaderless by the appropriate social orientation. As such, Assumption~\ref{as:leaderless} is widely applicable. The behavior of the natural resource dynamic over a leaderless network will be further studied in Section~\ref{sec:sim}, after the stability of the system has been established.
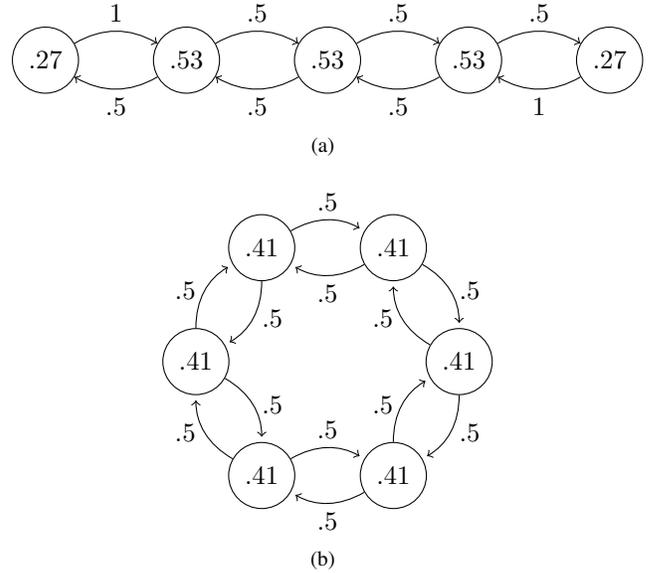
\begin{figure}[h]
\centering
	\subfloat[]{\label{subfig:line}	\begin{tikzpicture}
		[
		->,
		auto,
		node distance=0.5cm,
		every text node part/.style={align=center},
		scale=0.25, every node/.style={scale=1}
		]
		\tikzstyle{every state}=[fill=none,draw=black,text=black]
		\node[
		state
		] (n0) at(-15,0)
		{$.27$};
		\node[
		state
		] (n1) at(-7.5,0)
		{$.53$};
		\node[
		state
		] (n2) at(0,0)
		{$.53$};
		\node[
		state
		] (n3) at(7.5,0)
		{$.53$};
		\node[
		state
		] (n4) at(15,0)
		{$.27$};
		\path (n0) edge [->, bend left] node [above] {$1$}(n1)
		(n1) edge [->, bend left] node [below] {$.5$}(n0)
		(n1) edge [->, bend left] node [above] {$.5$}(n2)
		(n2) edge [->, bend left] node [below] {$.5$}(n1)
		(n2) edge [->, bend left] node [above] {$.5$}(n3)
		(n3) edge [->, bend left] node [below] {$.5$}(n2)
		(n3) edge [->, bend left] node [above] {$.5$}(n4)
		(n4) edge [->, bend left] node [below] {$1$}(n3);
		\end{tikzpicture}}\\
	\subfloat[]{\label{subfig:cycle}
		\begin{tikzpicture}
		[
		->,
		shorten >=2pt,
		auto,
		node distance=0.5cm,
		every text node part/.style={align=center},
		scale=0.25, every node/.style={scale=1}
		]
		\node[
		state
		] (n0) at(7,0)
		{$.41$};
		\node[
		state
		] (n1) at(3.5,6.0622)
		{$.41$};
		\node[
		state
		] (n2) at(-3.5,6.0622)
		{$.41$};
		\node[
		state
		] (n3) at(-7,0)
		{$.41$};
		\node[
		state
		] (n4) at(-3.5,-6.0622)
		{$.41$};
		\node[
		state
		] (n5) at(3.5,-6.0622)
		{$.41$};
		\path (n0) edge [->, bend left] node [left] {$.5$}(n1)
		(n1) edge [->, bend left] node [right] {$.5$}(n0)
		(n1) edge [->, bend left] node [below] {$.5$}(n2)
		(n2) edge [->, bend left] node [above] {$.5$}(n1)
		(n2) edge [->, bend left] node [right] {$.5$}(n3)
		(n3) edge [->, bend left] node [left] {$.5$}(n2)
		(n3) edge [->, bend left] node [right] {$.5$}(n4)
		(n4) edge [->, bend left] node [left] {$.5$}(n3)
		(n4) edge [->, bend left] node [above] {$.5$}(n5)
		(n5) edge [->, bend left] node [below] {$.5$}(n4)
		(n5) edge [->, bend left] node [left] {$.5$}(n0)
		(n0) edge [->, bend left] node [right] {$.5$}(n5);
		\end{tikzpicture}}
	\caption{$2$ Leaderless Networks: a line graph (\ref{subfig:line}) and a cycle graph (\ref{subfig:cycle}). Each edge is labeled with its edge weight $\omega_{ij}$ and each node is labeled with its social attribution $b_{i}v_{i}$.}\label{fig:led_ex}
\end{figure}

\subsection{Dynamics}
With these assumptions in place, the transformed network level dynamics will be derived. Computing the time derivative of $z$ gives
\begin{equation}
	\dot{z} = \frac{\dot{x}}{x} = 1 - x - \sum_{i=1}^{\n}y_i = 1 - e^z - u.
\end{equation}
Differentiating $u$ with respect to time and expanding gives
\begin{equation}
	\dot{u} = \sum_{i=1}^{\n} \dot{y}_i,
\end{equation}
\begin{align}
	\dot{u} &= \sum_{i=1}^{\n} \big(\b_i\upalpha_i x - \b_i\upalpha_i\uprho_i\big) - \sum_{i=1}^{\n}\sum_{j=1}^{\n}\b_i\upnu_i\w_{ij}(y_i - y_j) \\
	&= \sum_{i=1}^{\n}\big(\b_i\upalpha_i e^z - \b_i\upalpha_i\uprho_i\big) \\ &{}\qquad\qquad -\sum_{i=1}^{\n}\bigg[\sum_{j=1}^{\n}\big(\w_{ij}\b_i\upnu_i - \w_{ji}\b_j\upnu_j\big)\bigg]y_i \\
        &= \sum_{i=1}^{\n} \b_i \upalpha_i e^z - \sum_{i=1}^{\n} \b_i \upalpha_i - \sum_{i=1}^{\n}\b_i\upalpha_i(\uprho_i - 1),
\end{align}
\noindent where the sum containing $y_i$ has vanished due to Assumption~\ref{as:leaderless}. Continuing, 
\begin{align}
\begin{split}
	\dot{u} &= \bigg[\sum_{i=1}^{\n}\b_i\upalpha_i\bigg](e^z - 1) - \sum_{i=1}^{\n}\b_i\upalpha_i(\uprho_i - 1) \\
	&= \K_1(e^z - 1) - \K_2,
\end{split}
\end{align}
where
\begin{equation}
	\K_1 := \sum_{i=1}^{\n}\b_i\upalpha_i \textnormal{ and } \K_2 := \sum_{i=1}^{\n} \b_i\upalpha_i(\uprho_i - 1).
\end{equation}

From these definitions, $\K_1$ is manifestly positive because it is a sum of positive terms. Under Assumption~\ref{as:rhoi}, 
it also follows that
{\setlength{\mathindent}{0pt}
\begin{align*}
\begin{split}
	|\K_2| &= \left|\sum_{i=1}^{\n} \b_i\upalpha_i(\uprho_i - 1)\right| \leq \max_{i \in \{1,\dots,\n\}} |\uprho_i - 1| \sum_{i=1}^{\n} \b_i\upalpha_i \\
	&= \K_1\max_{i \in \{1, \ldots, n\}} |\uprho_i - 1| \leq \K_1, 
\end{split}
\end{align*}}
where the last inequality follows from Assumption~\ref{as:rhoi}.

The $(z, u)$ dynamics thus take the form
\begin{subequations} \label{eq:newsys}
\begin{align}
	\dot{z} &= 1 - e^z - u \\
	\dot{u} &= \K_1(e^z - 1) -\K_2.
\end{align}
\end{subequations}
Next, the equilibrium of the $(z, u)$ system is computed in order to translate
this system into one whose equilibrium is at the origin. 

\subsection{Equilibrium}
The following lemma provides the uniqueness and value of the $(z, u)$ system's equilibrium point. 

\begin{lemma}\label{lem:eq}
The $(z, u)$ system has a unique equilibrium point located at
\begin{align}
	\z_0 &= \log \left(\frac{\K_2}{\K_1} + 1\right) \\
	\u_0 &= -\frac{\K_2}{\K_1}.
\end{align}
\end{lemma}

\begin{proof}
Setting $\dot{u} = 0$ we find
\begin{equation}
	\dot{u} = \K_1(e^z - 1) - \K_2 = 0,
\end{equation}
which immediately provides
\begin{equation}
	\z_0 = \log(\K_2/\K_1 + 1). 
\end{equation}
Setting $\dot{z} = 0$ gives
\begin{equation}
	\dot{z} = 1 - e^z - u = 0,
\end{equation}
where setting $z = \z_0$ results in
\begin{equation}
\dot{z} = 1 - \left(\frac{\K_2}{\K_1} + 1\right) - u = 0.
\end{equation}
Solving for $\u_0$ then provides
\begin{equation}
	\u_0 = -\frac{\K_2}{\K_1}. 
\end{equation}
\end{proof}

By Lemma~\ref{lem:eq},  the equilibrium value of the resource, $\R_0$, (in the coordinates of Eq.~\eqref{eq:ses}) is
\begin{equation}
	\R_0 = \left(\frac{\K_2}{\K_1} + 1\right)\Rmax.
\end{equation}
If $\K_2 > 0$, which corresponds to at least one agent having $\Ri > \Rmax$, then $\R_0$ is larger than $\Rmax$. This occurs as the agents with $\Ri > \Rmax$ will work to increase $\R$ beyond $\Rmax$. Alternatively, if $\K_2$ is negative then $\R_0$ is smaller than $\Rmax$, though $\R_0$ is always positive because $|\K_2| < \K_1$ and thus $\K_2/\K_1 + 1 > 0$ always. Similarly, if $\K_2 > 0$, then $\u_0 > 0$, corresponding to a net effort to increase the available quantity of resource, while $\K_2 < 0$ causes $\u_0 < 0$, which corresponds to active resource consumption at steady state. 

Having computed the equilibrium point of the system, we define a coordinate shift by
\begin{align}
	v = z - \z_0, \quad w = u - \u_0,
\end{align}
resulting in the dynamics
\begin{align*}
\begin{split}
\dot{v} &= \dot{z} = 1 - e^z - u = 1 - e^{v + \z_0} - (w + \u_0), \\
	&= - e^v e^{\z_0} - w + 1 + \frac{\K_2}{\K_1}= -e^v e^{\z_0} - w + e^{\z_0}, \\
	&= -e^{\z_0}(e^v - 1) - w,
\end{split}
\end{align*}
where we have used $e^{\z_0} = \K_2/\K_1 + 1$. 

For $w$, the dynamics are governed by
\begin{align*}
\dot{w} &= \dot{u} = \K_1(e^z - 1) - \K_2 \\
        &= \K_1e^{v+\z_0} - \K_1 - \K_2 \\
        &= \K_1e^v\left(1 + \frac{\K_2}{\K_1}\right) - \K_1 - \K_2 \\
        &= \K_1e^v + \K_2e^v - \K_1 - \K_2 \\
        &= (\K_1 + \K_2)(e^v - 1). 
\end{align*}

The final system dynamics to be analyzed are
\begin{align} \label{eq:mainsys}
\dot{v} &= -e^{\z_0}(e^v - 1) - w \\
\dot{w} &= (\K_1 + \K_2)(e^v - 1),
\end{align}
whose unique equilibrium point is the origin. 

\subsection{Global Stability}
\label{sec:stability}
The following theorem demonstrates asymptotic stability
of the system in Equation~\eqref{eq:mainsys} to the origin. 

\begin{theorem} \label{thm:main}
Under Assumptions~\ref{as:leaderless} and \ref{as:rhoi} the origin is globally asymptotically stable in Equation~\ref{eq:mainsys}.
\end{theorem}
\begin{proof}
Consider the Lyapunov function 
\begin{equation}
	V(v, w) = e^v - v - 1 + \frac{(\K_1 + \K_2)^{-1}}{2}w^2, 
\end{equation}
which is positive definite, satisfies $V(0, 0) = 0$, and is radially unbounded.
Differentiating $V$ with respect to time, 
\begin{align}
\begin{split}
	\dot{V} &= e^v\dot{v} - \dot{v} + (\K_1 + \K_2)^{-1}w\dot{w} \\
	&= \begin{aligned}[t] &e^v(-e^{\z_0}(e^v - 1) - w) + e^{\z_0}(e^v - 1) \\ &+ w + w(e^v - 1)\end{aligned} \\
	&= \begin{aligned}[t]&-e^{\z_0}e^v(e^v - 1) - e^v w + e^{\z_0}(e^v - 1) \\&+ w + w(e^v - 1)\end{aligned} \\
	&= -e^{\z_0}(e^v - 1)^2 + w(e^v - 1) - w(e^v - 1) \\
	&= -e^{\z_0}(e^v - 1)^2 \leq 0.
\end{split}
\end{align}
Here, LaSalle's invariance principle can be used to prove global asymptotic stability of $(0, 0)$ by showing
that the set $V_{0}=\{(v, w) \mid \dot{V}(v, w) = 0\}$ contains only the trivial trajectory $\big(v(t), w(t)\big) \equiv (0, 0)$ \cite{khalil2002nonlinear}.

To do so, observe that $\dot{V}(v, w) = 0$ for all trajectories of the form $(0, w)$. 
Plugging this into the system dynamics in Equation~\eqref{eq:mainsys} implies
\begin{align*}
	\dot{v} = -w, \quad \dot{w} = 0.
\end{align*}
Then the only invariant point in $V_{0}$ has $w = 0$ because $\dot{v} = 0$ must hold to ensure that the system remains in $V_{0}$.  
\end{proof}

\section{Simulations}
\label{sec:sim}
This section considers the behavior of leaderless network topologies in simulation, focusing specifically on the case of the star graph. The star graph, also known as a hub, is of central importance to the study of complex networks \cite{newman2010networks}. The star graph also has a node, the center of the star, that might be expected to be the leader of a social network. Despite this, there are many leaderless social networks that can evolve over the star graph. Three leaderless networks on the same star topology are shown in Figure \ref{fig:3star}.

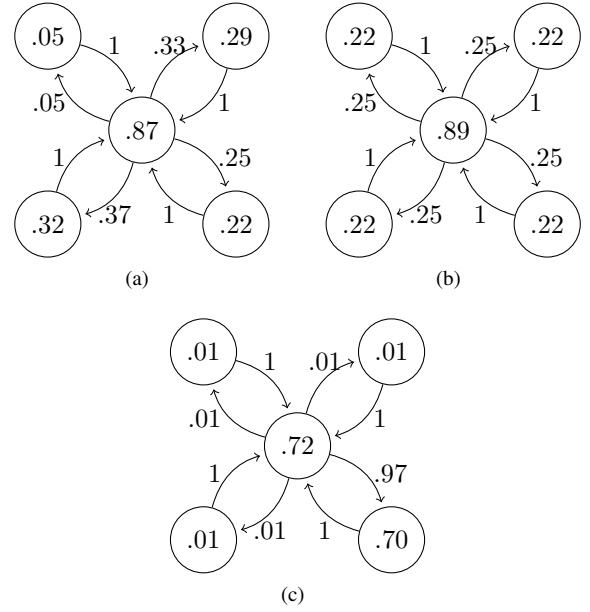
\begin{figure}[h]
\centering
	\subfloat[]{\label{subfig:1}	\begin{tikzpicture}
		[
		->,
		shorten >=2pt,
		auto,
		node distance=0.5cm,
		every text node part/.style={align=center},
		scale=0.25, every node/.style={scale=1}
		]
		\tikzstyle{every state}=[fill=none,draw=black,text=black]
		\node[
		state
		] (n0) at(-5,5)
		{$.05$};
		\node[
		state
		] (n1) at(-5,-5)
		{$.32$};
		\node[
		state
		] (n2) at(0,0)
		{$.87$};
		\node[
		state
		] (n3) at(5,5)
		{$.29$};
		\node[
		state
		] (n4) at(5,-5)
		{$.22$};

		\path (n2) edge [->, bend left] node [below,left] {$.05$}(n0)
		(n0) edge [->, bend left] node [above] {$1$}(n2)
		(n2) edge [->, bend left] node [below] {$.37$}(n1)
		(n1) edge [->, bend left] node [above,left] {$1$}(n2)
		(n2) edge [->, bend left] node [above] {$.33$}(n3)
		(n3) edge [->, bend left] node [below,right] {$1$}(n2)
		(n2) edge [->, bend left] node [above,right] {$.25$}(n4)
		(n4) edge [->, bend left] node [below] {$1$}(n2);
		\end{tikzpicture}}\hspace{.5cm}
	\subfloat[]{\label{subfig:2}
		\begin{tikzpicture}
		[
		->,
		shorten >=2pt,
		auto,
		node distance=0.5cm,
		every text node part/.style={align=center},
		scale=0.25, every node/.style={scale=1}
		]
			\node[
		state
		] (n0) at(-5,5)
		{$.22$};
		\node[
		state
		] (n1) at(-5,-5)
		{$.22$};
		\node[
		state
		] (n2) at(0,0)
		{$.89$};
		\node[
		state
		] (n3) at(5,5)
		{$.22$};
		\node[
		state
		] (n4) at(5,-5)
		{$.22$};
		\path (n2) edge [->, bend left] node [below,left] {$.25$}(n0)
		(n0) edge [->, bend left] node [above] {$1$}(n2)
		(n2) edge [->, bend left] node [below] {$.25$}(n1)
		(n1) edge [->, bend left] node [above,left] {$1$}(n2)
		(n2) edge [->, bend left] node [above] {$.25$}(n3)
		(n3) edge [->, bend left] node [below,right] {$1$}(n2)
		(n2) edge [->, bend left] node [above,right] {$.25$}(n4)
		(n4) edge [->, bend left] node [below] {$1$}(n2);
		\end{tikzpicture}}\\
		\subfloat[]{\label{subfig:3}
		\begin{tikzpicture}
		[
		->,
		shorten >=2pt,
		auto,
		node distance=0.5cm,
		every text node part/.style={align=center},
		scale=0.25, every node/.style={scale=1}
		]
		\tikzstyle{every state}=[fill=none,draw=black,text=black]
				\node[
		state
		] (n0) at(-5,5)
		{$.01$};
		\node[
		state
		] (n1) at(-5,-5)
		{$.01$};
		\node[
		state
		] (n2) at(0,0)
		{$.72$};
		\node[
		state
		] (n3) at(5,5)
		{$.01$};
		\node[
		state
		] (n4) at(5,-5)
		{$.70$};
	
		\path (n2) edge [->, bend left] node [below,left] {$.01$}(n0)
		(n0) edge [->, bend left] node [above] {$1$}(n2)
		(n2) edge [->, bend left] node [below] {$.01$}(n1)
		(n1) edge [->, bend left] node [above,left] {$1$}(n2)
		(n2) edge [->, bend left] node [above] {$.01$}(n3)
		(n3) edge [->, bend left] node [below,right] {$1$}(n2)
		(n2) edge [->, bend left] node [above,right] {$.97$}(n4)
		(n4) edge [->, bend left] node [below] {$1$}(n2);
		\end{tikzpicture}}
	\caption{$3$ Leader-Less Star Graphs for (\ref{subfig:1}) random weights, (\ref{subfig:2}) uniform weights, (\ref{subfig:3}) skewed weights. Each edge is labeled with its weight and each node is labeled with its social orientation.}\label{fig:3star}
\end{figure}

The resource was assumed to have a carrying capacity $\Rmax=1$, a growth rate $\r=1$, and a random initial condition that was fixed across simulations. The network was run with an ecological attribution $\a$ and a set of thresholds $\R$ where $$\a=\begin{bmatrix}
0.4340\\
0.2046\\
0.1891\\
0.6935\\
0.2108\\
\end{bmatrix}, ~~~~\R=\begin{bmatrix}
0.2262\\
0.4788\\
0.4582\\
1.1745\\
0.8483\\
\end{bmatrix}.$$ The time history of the resource level was identical for all $3$ systems and is shown in Figure \ref{fig:sys}. The individual usages for each of the $3$ systems are shown in Figure \ref{fig:3run}. 
\begin{figure}
\centering
\includegraphics[width=.75\columnwidth]{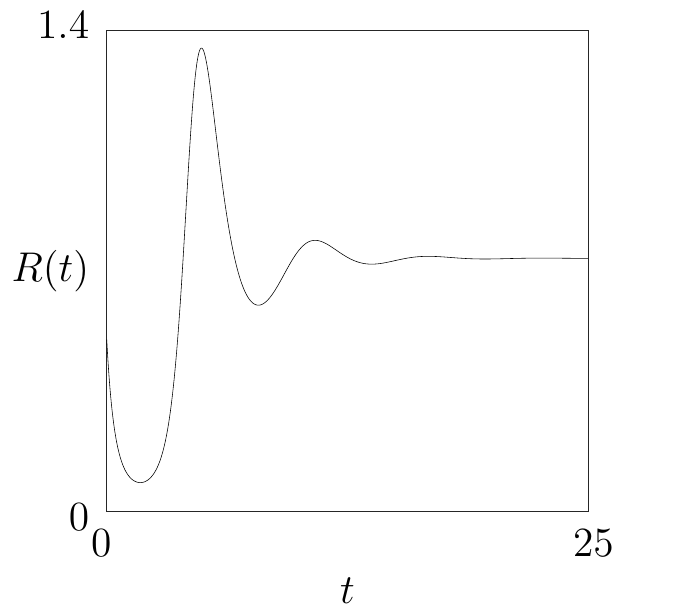}
\caption{Level of the Natural Resource over time of the three $5$ node leaderless star graphs shown in Figure \ref{fig:3star}.}\label{fig:sys}
\end{figure}

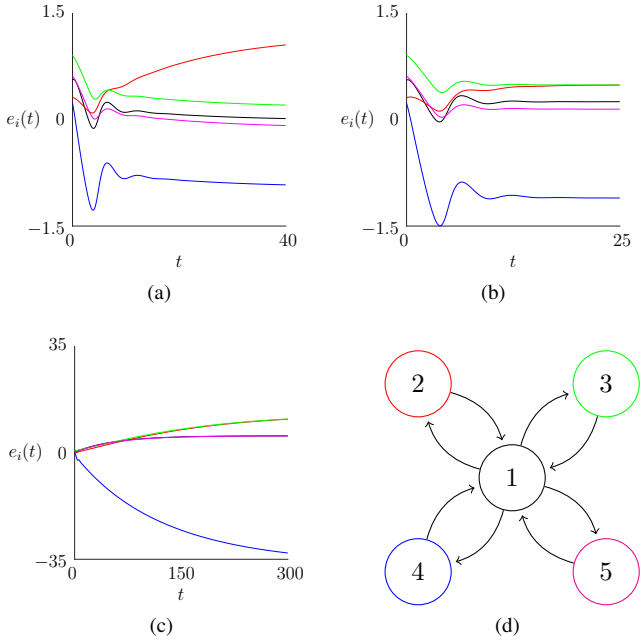
\begin{figure}[h!] \centering
\subfloat[]{\label{subfig:2_1}	\includegraphics[width=.5\columnwidth]{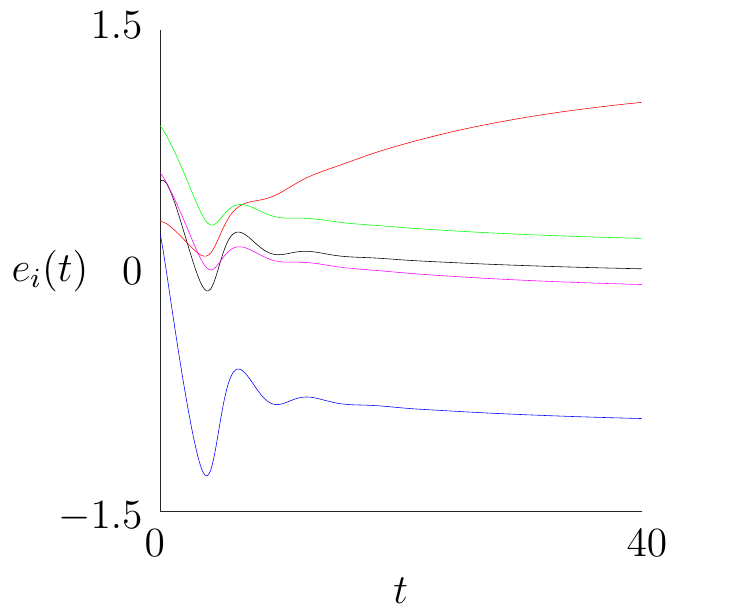}}
\subfloat[]{\label{subfig:2_2}	\includegraphics[width=.5\columnwidth]{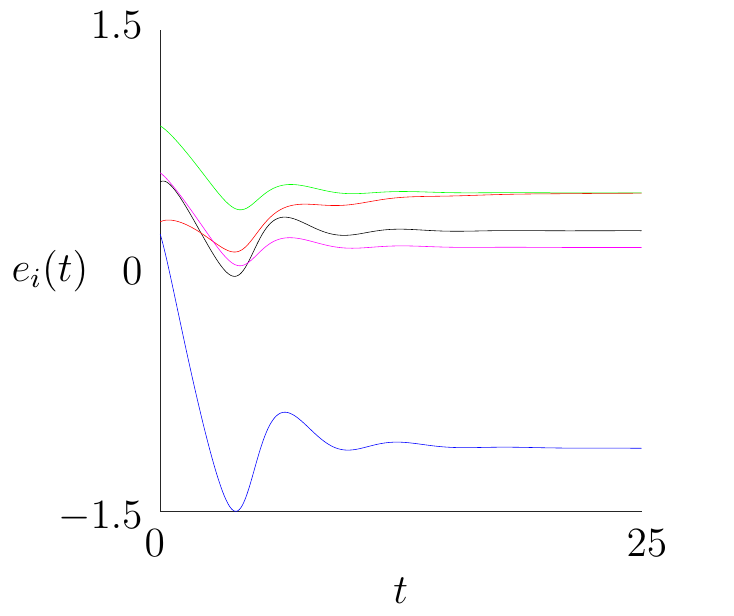}}\\
\subfloat[]{\label{subfig:2_3}	\includegraphics[width=.5\columnwidth]{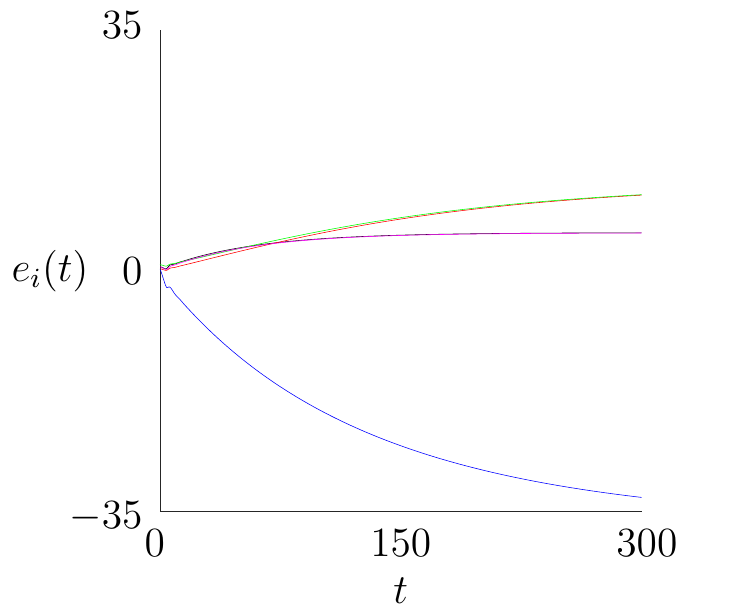}} \hspace{.5cm}
		\subfloat[]{\label{subfig:2_4}
		\begin{tikzpicture}
		[
		->,
		shorten >=2pt,
		auto,
		node distance=0.5cm,
		every text node part/.style={align=center},
		scale=0.25, every node/.style={scale=1}
		]
		\tikzstyle{every state}=[fill=none,draw=black,text=black]
				\node[
		state,draw=red
		] (n0) at(-5,5)
		{$2$};
		\node[
		state,draw=blue
		] (n1) at(-5,-5)
		{$4$};
		\node[
		state
		] (n2) at(0,0)
		{$1$};
		\node[
		state,draw=green
		] (n3) at(5,5)
		{$3$};
		\node[
		state,draw=magenta
		] (n4) at(5,-5)
		{$5$};
		\path (n2) edge [->, bend left] (n0)
		(n0) edge [->, bend left] (n2)
		(n2) edge [->, bend left] (n1)
		(n1) edge [->, bend left] (n2)
		(n2) edge [->, bend left] (n3)
		(n3) edge [->, bend left] (n2)
		(n2) edge [->, bend left] (n4)
		(n4) edge [->, bend left] (n2);
		\end{tikzpicture}}
\caption{Evolution of Individual Resource Consumption for $3$ Leader-Less Star Graphs: (\ref{subfig:2_1}) the random weighted graph shown in (\ref{subfig:1}). (\ref{subfig:2_2}) the uniform weighted graph shown in (\ref{subfig:2}). (\ref{subfig:2_3}) the skew weighted graph shown in (\ref{subfig:3}). (\ref{subfig:2_4}) maps the position of the nodes to trajectories}\label{fig:3run}
\end{figure}

\section{Discussion}\label{sec:disc}

For each leaderless network considered here the aggregate resource consumption is identical, however the individual usage changes dramatically for the various weighting schemes. Further, given that the form of the individual consumption effort given in Eq. \eqref{eq:eff} is similar to the consensus dynamic\cite{abelson1964mathematical,mesbahi2010graph}, it is reasonable to expect that individual consumption converges to a single uniform steady-state usage as happens under consensus on a connected graph. This does not occur in leaderless consumption networks, showing that this model addresses a need for social models that do not reach consensus \cite{friedkin2015problem}. 

To see why this behavior occurs, first note that the model allows negative resource usage and that stability of the equilibrium point requires that an agent (here agent $4$ as shown in Figure \ref{fig:3run}) contributes resource to ensure balance with the usage of the other agents. As the network is leaderless, the equilibrium behavior of the system depends on the ecological component. Agent $4$ has a threshold, $R_{4}=1.17$, which is significantly higher than the thresholds of its neighbors. This higher threshold drives the agent to contribute resource to balance out the usage of the agents that have lower thresholds and which therefore consume the resource.

While this system is stable, as shown by Theorem \ref{thm:main} and displayed in Figure \ref{fig:sys}, this system level behavior would be worrying as the progress of a natural resource. Imagine, for example, the panic of a populace if the level of the local water reservoir were to change as indicated in Figure \ref{fig:sys}: The reservoir shifts quickly from being almost empty to overflowing and then starts heading back down towards empty before reaching equilibrium. 

The issue raised in the preceding paragraph points to the fact that stability, while vital for understanding the behavior of a system, is not the only property of a natural resource system which must be understood. There are other questions, those related to sustainability, which must be addressed about these models before they are used to inform decision making in resource governance problems. For example, can humans use this resource in the short term without risking the depletion of the resource in the long term? Future work is required to bridge this gap between stability tools and the characterization of an ecological system as sustainable. 

\section{Conclusion}\label{sec:conc}
In this paper, a model of natural resource consumption was considered. Stability was shown under the assumption that the network is leaderless and that the individual thresholds are in $(0,2\Rmax)$. The stability of the system was verified in simulation and it was shown that while the system is stable it is not necessarily sustainable. 
\bibliographystyle{unsrt}
\bibliography{waterresourcebib}

\addtolength{\textheight}{-12cm}   




\end{document}